\documentclass[runningheads]{llncs}

\usepackage[utf8]{inputenc}
\usepackage[american]{babel}
\usepackage[T1]{fontenc}
\usepackage{amsmath}
\usepackage{amssymb}

\usepackage{amsthm}

\usepackage[pdftex]{graphicx}

\usepackage[disable]{todonotes}
\usepackage{url}

\usepackage{xspace}

\usepackage{hyphenat}

\newcommand{\ie}[0]{i.\,e.\@\xspace}

\newcommand{\inmsg}[2]{#1?#2}
\newcommand{\outmsg}[2]{#1!#2}
\newcommand{\inputs}{\mathcal{I}}

\newcommand{\msg}[2]{#1.#2}
\newcommand{\trans}[1]{\xrightarrow{#1}}
\newcommand{\concat}{\smallfrown}
\newcommand{\emptytrace}{\varnothing}

\newcommand{\traces}{\mathcal{T}}

\newcommand{\fand}{\sqcap}
\newcommand{\for}{\sqcup}
\newcommand{\fimp}{\Rightarrow}
\newcommand{\fnot}[1]{\overline{#1}}
\newcommand{\faulteq}{\approx}
\newcommand{\irrelevant}{\square}
\newcommand{\texists}{\mathit{ex}}
\newcommand{\tvalue}{\mathit{val}}

\newcommand{\powerset}[1]{\mathcal{P}(#1)}

\begin{document}
\title{Compositionality of Component Fault Trees -- Definitions and Proofs}
\author{Simon Greiner \and 
Peter Munk \and 
Arne Nordmann } 
\authorrunning{S. Greiner et al.}
%
\institute{Robert Bosch GmbH, Corporate Sector Research, Renningen, Germany
\email{\{firstname.lastname\}@bosch.com}}
\maketitle     
\begin{abstract}
In the paper \emph{Compositionality of Component Fault Trees} \cite{greiner19}, we present a discussion of the compositionality of correctness of component fault trees.
In this technical report, we present the formal proof of the central theorem of the aforementioned publication.
\end{abstract}

\section{Introduction}
This technical report presents the formal proof for the compositionality of correctness of Component Fault Trees (CFTs).
This report accompanies the paper~\cite{greiner19}, which contains additional explanations and examples for the definitions, Lemmas and theorems presented here.
Please note that this technical report is not meant to be self contained, but is intended to be comprehensible only in combination with the original paper.

The structure of this report is as follows:
First we provide a formal definition for components, Component Fault Trees, and correctness of CFTs.
In the following section, we define composition of component fault trees and formulate our central theorem stating that correctness of CFTs is compositional.
The formal proof for this theorem is rather lengthy and technical and is presented in the final section.

\section{Components and Component Fault Trees}

In the remainder of this work we take the formalization of components from Greiner and Grahl~ \cite{greiner2016noninterference} and reuse their notation for better comparability of further results in this paper.

A component has an internal state, input ports, and output ports.
A component can receive messages via input ports and send messages via its output ports. Received messages can trigger the component to change its internal state.
Formally, we consider components as Input-Output Labeled Transition Systems (see \cite{rafnssonHS12} for a formal definition of IOLTS).
A port has a name and a signature, \ie names and types of variables that can be communicated via the port.
For a message $m$ communicated via an input port with name $p$ with value $v$, we write $m = \inmsg{p}{v}$, 
for a message $n$ communicated via an output port with name $q$ with value $w$, we write $n = \outmsg{q}{w}$.
We refer to the set of all messages communicated via an input port as inputs, and the set of all messages communicated via an output port as outputs.
If it does not matter whether a message is an input or an output, we write $m = \msg{p}{v}$ or $n = \msg{q}{w}$ respectively.
We write $c \trans{m} c'$ for a component $c$ communicating message $m$ and transitioning to a component $c'$. 
You can consider $c'$ to have a changed internal state.
If $c'$ is irrelevant, we write $c \trans{m}$, if there exists some $c'$ such that $c \trans{m} c'$.

The behavioral definition of a component limits the sequence of messages a component can communicate.
We refer to a sequence of messages as a trace. 
We use $\concat$ as the concatenation operator for traces and $\emptytrace$ to refer to the empty trace.
The length of a trace is defined as the amount of messages in a trace.
We write $c \trans{t} c'$ if a component $c$ transitions to component $c'$ while communicating the trace $t$.
A component $c$ can communicate a trace $t \concat m$ while transitioning to component $c'$, if there exists a component $c''$ such that $c \trans{t} c''$ and $c'' \trans{m} c'$.
We again write $c \trans{t} $, if there exists some $c'$ such that $c \trans{t} c'$.
We refer to all possible traces as $\traces$.
Finally, we explicitly define environments in which components can run. Environments model the entities providing inputs to a component after observing the behavior of the component.
\begin{definition}[Environment]
	\label{def:environment}
	An environment $\omega$ is a function $\traces \mapsto \powerset{\inputs}$, where $\powerset{\inputs}$ is the powerset of all inputs.	
\end{definition}

\begin{definition}[Communication under Environment]
	A component $c$ can communicate a trace $t$ under an environment $\omega$, written $\omega \models c\trans{t}$, iff $c \trans{t}$ and for all $t_1, t_2, \inmsg{p}{v}$ with $t = t_1 \concat p?v \concat t_2$, it holds that $\inmsg{p}{v} \in \omega(t_1)$.	
\end{definition}

\begin{definition}[Event]
	\label{def:event}
	An event $E$ is a tuple $(p,t)$, where $p$ is port and $t$ is a type.
	If $p$ is an input (output) port, $E$ is an input (output) event.
\end{definition}

\begin{definition}[CFT]
	\label{def:cft}
	A CFT is a tuple $(P,E)$, where $P$ is a propositional logic formula, where each literal in $P$ is an input event, and each input event only appears non-negated in $P$. $E$ is an output event.
\end{definition}

For every propositional formula $P$, with literals only appearing non-negated, the formula $\fnot{P}$ is a propositional formula with literals only appearing negated.
For $\fnot{P}$ we can always find a disjunctive normal form with clauses $\fnot{P_1}, ... \fnot{P_n}$, such that $\fnot{P} = \fnot{P_1} \for ... \for \fnot{P_n}$.
$\fnot{P} \fimp \fnot{E}$ then holds, iff for all clauses $\fnot{P}_i$ it holds that $\fnot{P}_i \fimp \fnot{E}$.
\begin{definition}[Clause]
	A clause $\fnot{P_i}$ is a propositional formula, where each literal only appears negated and $\fand$ is the only logical operator in the formula.
\end{definition}	

\begin{definition}[Message Event-Equivalence]
	\label{def:message:equiv}
	A message $m=q.v$ is irrelevant w.r.t. an event $A_i = (q_i, t_i)$, if $q \neq q_i$.
	For a message $m$ that is irrelevant, we write $m \faulteq_{\fnot{A_i}} \irrelevant$.
	
	Two messages $m_1 = q_1.v_1$ and $m_2 = q_2.v_2$ are event-equivalent w.r.t. an event $A_i = (q_i, t_i)$, written $m_1 \faulteq_{\fnot{A_i}} m_2$, if
	\begin{align*}
		& m_1 \faulteq_{\fnot{A_i}} \irrelevant \text{ and } m_2 \faulteq_{\fnot{A_i}} \irrelevant  \text{ or }\\
		& q_i = q_1 = q_2 \text{ and } t_i = \texists  \text{ or }\\
		& q_i = q_1 = q_2 \text{ and } t_i = \tvalue \text{ and } v_1 = v_2
	\end{align*}
	Two messages $m_1 = q_1.v_1$ and $m_2 = q_2.v_2$ are event-equivalent w.r.t. a clause $\fnot{P_i} = \fnot{A_1} \fand \dots \fand \fnot{A_n}$, written $m_1 \faulteq_{\fnot{P_i}} m_2$, if
	$m_1 \faulteq_{\fnot{A_i}} m_2 \text{ for all } 0 < i \leq n$
	
	Two messages $m_1 = q_1.v_1$ and $m_2 = q_2.v_2$ are event-equivalent w.r.t. a clause $\fnot{P_i}$ and an event $E$, written $m_1 \faulteq_{\fnot{P_i},\fnot{E}} m_2$, if
	$ m_1 \faulteq_{\fnot{P_i}} m_2 \text{ and } m_1 \faulteq_{E} m_2$
\end{definition}

\begin{definition}[Trace Event-Equivalence]
	\label{def:trace:equiv}
	Two traces $t_1, t_2$ are event-equivalent w.r.t. an event $A_i$, written $t_1 \faulteq_{\fnot{A_i}} t_2$, iff
	\begin{align*}
		& t_1 = \emptytrace \text{ and } t_2 = \emptytrace \text{ or} \\
		& t_1 = m_1 \concat t_1' \text{ and } m_1 \faulteq_{\fnot{A_i}} \irrelevant \text { and } t_1' \faulteq_{\fnot{A_i}} t_2 \text{ or} \\
		& t_2 = m_2 \concat t_2' \text{ and } m_2 \faulteq_{\fnot{A_i}} \irrelevant \text { and } t_1 \faulteq_{\fnot{A_i}} t_2' \text{ or} \\
		& t_1 = m_1 \concat t_1' \text{ and } t_2 = m_2 \concat t_2' \text{ and } m_1 \faulteq_{\fnot{A_i}} m_2 \text{ and } t_1' \faulteq_{\fnot{A_i}} t_2' 
	\end{align*}
	Event-equivalence of traces w.r.t. a clause is defined analogously.
\end{definition}

\begin{definition}[Erroneous Environment]
	\label{def:errorenv}	
	$\omega_f$ is an erroneous environment for $\omega$ w.r.t. a clause $\fnot{P_i}$ and an output event $E$, written $\omega_f \faulteq_{\fnot{P_i},\fnot{E}} \omega$, iff for all $t_1 \faulteq_{\fnot{P_i},\fnot{E}} t_2$ it holds that
	\begin{align}
		& \forall \inmsg{p}{v} \in \omega(t_1) && \bullet (\inmsg{p}{v} \faulteq_{\fnot{P_i},\fnot{E}} \irrelevant \; \text{ or } \; \exists \inmsg{q}{u} \in \omega_f(t_2) \bullet \inmsg{p}{v} \faulteq_{\fnot{P_i}} \inmsg{q}{u}) \text{ and} \label{def:errorenv:equalin}\\
		& \forall \inmsg{q}{u} \in \omega_f(t_2) && \bullet (\inmsg{q}{u} \faulteq_{\fnot{P_i},\fnot{E}} \irrelevant \; \text{ or } \; \exists \inmsg{p}{v} \in \omega(t_1)  \bullet \inmsg{p}{v} \faulteq_{\fnot{P_i}} \inmsg{q}{u}) \label{def:errorenv:noadditional}
	\end{align}
\end{definition}

\begin{definition}[Environment (extd.)]
	\label{def:environment:extd}
	A function $\omega$ is an environment w.r.t. a clause $\fnot{P_i}$ and an event $E$, iff it is an environment according to Definition \ref{def:environment}, and $\omega$ is an erroneous environment to itself according to Definition \ref{def:errorenv} (\ie, $\omega \faulteq_{\fnot{P_i},\fnot{E}} \omega$).
\end{definition}

\begin{definition}[Clause correctness w.r.t. a component]
	\label{def:safety:clause}
	Given the relation $\faulteq_{\fnot{P},\fnot{E}}$ for the clause $\fnot{P}$ and the output event $E$.
	$\fnot{P}$ and $E$ are correct w.r.t. a component $c$, if
	\begin{align*}
		\forall \omega_f, \omega_c \forall t_f \bullet \; \omega_f \faulteq_{\fnot{P},\fnot{E}} \omega_c \; \wedge \;  \omega_f  \models c \trans{t_f} \; \implies \;
		\exists t_c \bullet \omega_c \models c \trans{t_c} \; \wedge \; t_f \faulteq_{\fnot{P},\fnot{E}} t_c
	\end{align*}
\end{definition}

\begin{definition}[CFT correctness w.r.t. a component]
	\label{def:safety:cft}
	Given a component $c$, a CFT $(P,E)$ with $\fnot{P} = \fnot{P_1} \for ... \for \fnot{P_n}$.
	$(P,E)$ is correct w.r.t. $c$, iff $\fnot{P_i}$ and $E$ are correct w.r.t. $c$ for all $0 < i \leq n$.
\end{definition}

\section{Compositionality of Component Fault Tree Correctness}

\begin{definition}[CFT composition]
	\label{def:cft:composition}
	Let $c$ and $d$ be components, $p_1 \ldots p_n$ ports, which are input ports of $c$ and outputs ports of $d$, $A^1_1 \ldots A^1_{m_1}, \ldots A^n_{1} \ldots A^n_{m_n}$ events with $A^i_j$ being an event on port $p_i$, i.e. an input event of $c$ and an output event of $d$.
	Let further be $(P_c,E)$ be a CFT of $c$ and $(P^i_j, A^i_j)$ the CFTs of $d$ for the output events $A^i_j$.
	
	We define the CFT of the composition of $c$ and $d$ for output event $E$ as $(P_{c,d},E)$, where
	$P_{c,d}$ is formula $P_c$ with each occurrence of $A^i_j$ is replaced by $P^i_j$.
\end{definition}

\begin{definition}[Deterministic components]
	\label{def:deterministic}
	A component $c$ is deterministic, if
	\begin{align*}
		& c \trans{p?v} &&\implies c \trans{p?v'} \text{ for all } v \text{ and } v' \text{ and}\\
		& c \trans{m_1} \text{ and } c \trans{m_2} \text{ and } m_1 \neq m_2 &&\implies m_1 \text{ and } m_2 \text{ are inputs, and} \\
		& c \trans{m} c_1 \text{ and } c \trans{m} c_2 &&\implies c_1 = c_2
	\end{align*}
\end{definition}
In the remainder, we assume all components to be deterministic.

\begin{theorem}[Composed CFT correctness w.r.t. composition]
	\label{thm:compositioncorrectness}
	Given components $c$ and $d$, CFT $(P_c, E)$ of $c$ and $(P^i_j, A^i_j)$ of $d$ as in Definition \ref{def:cft:composition}, such that the CFTs are correct w.r.t. the respective components.
	Then the composed CFT $(P_{c,d}, E)$ is correct w.r.t. the composition of $c$ and $d$.
\end{theorem}

\section{Proof of Compositionality of CFT Correctness}
\label{sec:proof1}

In this section, we present the proof for Theorem \ref{thm:compositioncorrectness}.

The structure of the proof is as follows.
First we define a notion of CFT composition in Definition \ref{def:cft:composition:strict} which is more strict than the composition as defined above.
We then show two lemmas about the structure of the formula resulting from the strict CFT composition in Lemmas \ref{lem:pidiff} and \ref{lem:clause:complete}.
We then show correctness of the strict composed CFT w.r.t. the component $c$ (Lemma \ref{lem:correct:c}) and $d$ (Lemma \ref{lem:correct:d}).
For the correctness w.r.t. $d$, we require an additional Lemma about counterexamples for the correctness (Lemma \ref{lem:counter:simple}).

Finally, the proofs for the correctness of the strict composition w.r.t. the composition of $c$ and $d$ (Theorem \ref{thm:cft:composition:strict}) and finally Theorem \ref{thm:compositioncorrectness} follow relatively directly from the Lemmas.

\begin{definition}[Strict CFT composition]
	\label{def:cft:composition:strict}
Let $c$ and $d$ be components, $p_1 \ldots p_n$ ports, which are input ports of $c$ and outputs ports of $d$, $A^1_1 \ldots A^1_{m_1}, \ldots A^n_{1} \ldots A^n_{m_n}$ events with $A^i_j$ being an event on port $p_i$, i.e. an input event of $c$ and an output event of $d$.
Let further be $(P_c,E)$ be a CFT of $c$ and $(P^i_j, A^i_j)$ the CFTs of $d$ for the output events $A^i_j$.

We then define the CFT of the strict composition of $c$ and $d$ for the output event $E$ as $(P_{c,d}',E)$, where
$P_{c,d}$ is the formula $P_c$ with each occurrence of $A^i_j$ is replaced by $A^i_j \fand P^i_j$.
\end{definition}

Note that Definition \ref{def:cft:composition:strict} differs from Definition \ref{def:cft:composition} only in a way such that we keep the events of the connected ports in the formula.


\begin{lemma}
	\label{lem:pidiff}
	Let $\fnot{P_1} \ldots \fnot{P_n}$ be the clauses of the CFT $(P_c, E)$, 
	\ie $\fnot{P_c} = \fnot{P_1} \for \ldots \for \fnot{P_n}$
	
	Then $\fnot{P_{c,d}'} = \fnot{P_1'} \for \ldots \for \fnot{P_n'}$ with 
	$\fnot{P_i'} = \fnot{P_i} \fand \fnot{P_{m+1}} \fand \ldots \fand \fnot{P_{m+k}}$, where $P_{m+i}$ is the formula from the CFT $(P_{m+1},A_i)$ of $d$.
\end{lemma}
\begin{proof}
	\begin{align}
	\fnot{P_{c,d}'} =& \fnot{P_1'} \for \ldots \for \fnot{P_n'} \text{ with} \\
	\fnot{P_i'} =& \fnot{A_1} \fand \ldots \fand \fnot{A_m} \fand 
	(\fnot{A_{m+1}} \fand \fnot{P_{m+1}}) \fand \ldots \fand (\fnot{A_{m+k}} \fand \fnot{P_{m+k}}) \\
	\intertext{ where  $\fnot{A_1} \fand \ldots \fand \fnot{A_{m+k}}$ is the clause $P_i$ of the formula $P_c$ of the CFT $(P_c, E)$.
		$A_i$ are input events of $c$, $A_{m+i}$ are input events of $c$ which are bound by composition to an output events of $d$, and $P_{m+i}$ is the formula from the CFT $(P_{m+i},A_i)$ of $d$.
		and due to symmetry of $\fand$ we get}
	\fnot{P_i'} =& \fnot{A_1} \fand \ldots \fand \fnot{A_m} \fand 
	\fnot{A_{m+1}} \fand \ldots \fand \fnot{A_{m+k}} \fand 
	\fnot{P_{m+1}} \fand \ldots \fand \fnot{P_{m+k}}) \\
	\intertext{By definition of the clause $P_i$, we get}
	\fnot{P_i'} =& \fnot{P_i} \fand \fnot{P_{m+1}} \fand \ldots \fand \fnot{P_{m+k}} \label{eq:pidiff}
	\end{align}
\end{proof}

\begin{lemma}
	\label{lem:clause:complete}
	Let $\fnot{P_1} \ldots \fnot{P_n}$ be the clauses of the CFT $(P_c, E)$, 
	\ie $\fnot{P_c} = \fnot{P_1} \for \ldots \for \fnot{P_n}$
	
	Then $\fnot{P_{c,d}'} = \fnot{P_1'} \for \ldots \for \fnot{P_n'}$ with 
	\begin{align}
	\fnot{P_i'} =& (\fnot{A_1} \fand \ldots \fand \fnot{A_m} \fand \fnot{A_{m+1}} \fand \fnot{P^{1x1}_{m+1}} \fand \ldots \fand \fnot{A_{m+k}} \fand \fnot{P^{1xk}_{m+k}}) \for \\
	&(\fnot{A_1} \fand \ldots \fnot{A_m} \fand \fnot{A_{m+1}} \fand \fnot{P^{2x1}_{m+1}} \fand \ldots \fand \fnot{A_{m+k}} \fand \fnot{P^{2xk}_{m+k}}) \for \\
	& \ldots \for \\
	& (\fnot{A_1} \fand \ldots \fnot{A_m} \fand \fnot{A_{m+1}} \fand \fnot{P^{ox1}_{m+1}} \fand \ldots \fand \fnot{A_{m+k}} \fnot{P^{oxk}_{m+k}})
	\end{align}
where $\fnot{A_1} \ldots \fnot{A_m}$ are clauses of the CFT $(P_c,E)$ and $\fnot{P^{1x1}_{m+1}}$ is a clause of the CFT $(P_{m+i}, A_{m+i})$ of $c$.
\end{lemma}
\begin{proof}
	We continue the re-ordering of the formula of $P_i'$ from the proof of Lemma \ref{lem:pidiff}
	\begin{align}
	\fnot{P_i'} =& \fnot{P_i} \fand \fnot{P_{m+1}} \fand \ldots \fand \fnot{P_{m+k}} \\
	\intertext{We can now replace the formulas $\fnot{P_{m+1}}$ by their clauses.
		We refer to the clauses of $\fnot{P_{m+1}}$ with $\fnot{P^1_{m+1}} \ldots \fnot{P^o_{m+1}}$. 
		Note that each of these clauses only contains input events of $d$. For easier presentation, we assume here w.l.g. that each CFT of $d$ defines $o$ clauses.}
	\fnot{P_i'} =& \fnot{P_i} \fand 
	(\fnot{P^1_{m+1}} \for \fnot{P^o_{m+1}}) \fand 
	(\fnot{P^1_{m+2}} \for \fnot{P^o_{m+2}}) \fand     			
	\ldots \fand
	(\fnot{P^1_{m+k}} \for \fnot{P^o_{m+k}})
	\intertext{Aplying distributive law on this formula, yields the following formula:}
	\fnot{P_i'} =& (\fnot{P_i} \fand \fnot{P^{1x1}_{m+1}} \fand \ldots \fnot{P^{1xk}_{m+k}}) \for 
	(\fnot{P_i} \fand \fnot{P^{2x1}_{m+1}} \fand \ldots \fnot{P^{2xk}_{m+k}}) \for \\
	& \ldots \for
	(\fnot{P_i} \fand \fnot{P^{ox1}_{m+1}} \fand \ldots \fand \fnot{P^{oxk}_{m+k}}) 
	\intertext{Again, we can restructure the formula}
	\fnot{P_i'} =& (\fnot{A_1} \fand \ldots \fnot{A_m} \fand \fnot{A_{m+1}} \fand \fnot{P^{1x1}_{m+1}} \fand \ldots \fand \fnot{A_{m+k}} \fand \fnot{P^{1xk}_{m+k}}) \for \\
	&(\fnot{A_1} \fand \ldots \fnot{A_m} \fand \fnot{A_{m+1}} \fand \fnot{P^{2x1}_{m+1}} \fand \ldots \fand \fnot{A_{m+k}} \fand \fnot{P^{2xk}_{m+k}}) \for \\
	& \ldots \for \\
	& (\fnot{A_1} \fand \ldots \fnot{A_m} \fand \fnot{A_{m+1}} \fand \fnot{P^{ox1}_{m+1}} \fand \ldots \fand \fnot{A_{m+k}} \fnot{P^{oxk}_{m+k}})    
	\end{align}
	
\end{proof}

Now we can show that the strict composition of the CFTs is correct w.r.t. the component $c$.
\begin{lemma}
	\label{lem:correct:c}
	If CFT $(P_c,E)$ is correct w.r.t. $c$, then $(P_{c,d}',E)$ is also correct w.r.t. $c$.
\end{lemma}
\begin{proof}
	A closer look at Lemma \ref{lem:pidiff} shows that the clause $\fnot{P_i'}$ differs from $\fnot{P_i}$ only by the clauses $\fnot{P_{m+1}}$ to $\fnot{P_{m+k}}$, which only contain events of input ports of $d$.
	By assumption, $c$ can not communicate messages over these ports, therefore the difference has no effect on the specification of $\faulteq_{\fnot{P_i'}, \fnot{E}}$ w.r.t. messages $c$ can communicate.
	Therefore, the correctness of $(P_{c,d}',E)$ w.r.t. $c$ follows directly from the correctness of $(P_c,E)$.
\end{proof}

\begin{lemma}
\label{lem:counter:simple}
	Let $\omega_f, \omega_c'', t_1$ be a counterexample for correctness of clause $\fnot{P_f}$ w.r.t. $d$ and $\omega_c''$ does not provide irrelevant inputs w.r.t. $\fnot{P_f},\fnot{E}$.
	
	Let further $\omega_c'$ and $\omega_f'$ be environments such that
	for all $t$ with for all prefixes $t_1'$ of $t_1$, \ie $t_1 = t_1' \concat t_1''$ it does not hold that $t \faulteq_{\fnot{P_f},\fnot{E}} t_1'$, it holds
	\begin{align*}
	\omega_f'(t) &= \omega_f(t) \setminus \{m | !(m \faulteq_{P_f,E} \irrelevant) \} \\
	\omega_c'(t) &= \omega_c''(t) \setminus \{m | !(m \faulteq_{P_f,E} \irrelevant) \}
	\end{align*}
	
	Let further $\omega_c'$ and $\omega_f'$ be such that for all $t$ with for all $t_1 = t_1' \concat p?v \concat t_1''$ and $t \faulteq_{\fnot{P_f},\fnot{E}} t_1'$ (\ie there exists a prefix of $t_1$ equivalent to $t$), 
	it holds
	\begin{align*}
	\omega_f'(t) &= \{ m | m \faulteq_{\fnot{P_{m+i}},\fnot{A_{m+i}}} \irrelevant \vee m = p?v\} \\
	\omega_c'(t) &= \{m\} \text{ for one $m$ with } m \in \omega_c''(t) \wedge m \faulteq_{\fnot{P_f},\fnot{E}} p?v
	\end{align*}
	
	Then $\omega_f', \omega_c', t_1$ is a counterexample for correctness of clause $\fnot{P_f}$ w.r.t. $d$.
\end{lemma}

Note that $\omega_f'$ only offers input messages used in the counterexample and $\omega_c'$ only offers messages required to be equivalent to $\omega_f'$.
Further, of course, $\omega_c'$ does not provide irrelevant input messages. 

\begin{proof}
	By definition, $\omega_c'$ only contains one message per trace and port, and does not provide invisible input.
	
	By construction, $\omega_c'$ and $\omega_f'$ are environments w.r.t. $\faulteq_{\fnot{P_f},\fnot{E}}$ (Definition \ref{def:environment:extd}).
	
	It holds that $\omega_f' \faulteq_{\fnot{P_f},\fnot{E}} \omega_c'$, $\omega_f' \models d \trans{t_1}$, and for all traces $t$ it holds $\omega_c' \models d \trans{t} \implies \omega_c \models d \trans{t}$.
	
	Therefore, if $\omega_f, \omega_c'', t_1$ is a counterexample for the correctness of $P_f,E$ w.r.t. $d$,
	then $\omega_f', \omega_c', t_1$ is a counterexample.
\end{proof}

Showing correctness of $(P_{c,d}',E)$ w.r.t. $d$ is more complex.
\begin{lemma}
	\label{lem:correct:d}
	If all CFTs $(P_{m+i},A_i)$ are correct w.r.t. $d$, then $(P_{c,d}',E)$ is correct w.r.t. $d$.
\end{lemma}
\begin{proof}
	All CFTs of $d$ are correct w.r.t. $d$.
	We assume towards contradiction that CFT $(P_{c,d}', E)$ is not correct w.r.t. $d$.
	
	Thus, there exist a clause $\fnot{P_f}$ defined by the CFT $(P_{c,d}', E)$, such that $\fnot{P_f}$ is not correct w.r.t $d$ according to Definition \ref{def:safety:cft}.
	
	According to Definition \ref{def:safety:clause}, there has to exists a correct environment $\omega_c$ and an erroneous environment $\omega_f$ with $\omega_f \faulteq_{\fnot{P_f},\fnot{E}} \omega_c$, and
	a trace $t_1$ with $\omega_f \models c \trans{t_1}$ such that for all $t_2$ with $\omega_c \models d \trans{t_2}$, 
	it does not hold that $t_1 \faulteq_{\fnot{P},\fnot{E}} t_2$.
	
	Note that $t_1$ and $t_2$ can only contain messages, which are communicated via ports which are input or output ports of $d$, since according to Definition \ref{def:safety:clause}, $d$ can communicate the traces.

	From the counterexample $\omega_f, \omega_c$, and $t_1$, we construct a simpler counterexample.
	We construct an environment $\omega_c''$ by removing all invisible events:
	$\omega_c''(t) = \{m | m \in \omega_c(t) \wedge !(m \faulteq_{\fnot{P_f},\fnot{E}} \irrelevant)\}$ for all $t$.
	According to Lemma 7 in \cite{greiner2016noninterference}, $\omega_f, \omega_c'', t_1$ is a counterexample for correctness of $P_f$ w.r.t. $d$.

	According to Lemma \ref{lem:counter:simple}, there exists a simpler counterexample $\omega_f'$, $\omega_c'$, $t_1$.

	We now show that  $\omega_f', \omega_c', t_1$ can not be a counterexample for correctness.
	
	Let $t_1'$ be the longest prefix of $t_1$, such that there exists a $t_2'$ under $\omega_f'$ and $\omega_c'$ with $t_1' \faulteq_{\fnot{P_f},\fnot{E}} t_2'$.
	Then $t_1' \concat m_1$ is a prefix of $t_1$.
	Further, $m_1$ has to be relevant w.r.t. $\faulteq_{\fnot{P_f},\fnot{E}}$, because otherwise $t_1' \concat m_1$ would be longer than $t_1'$ and $t_1' \concat m_1 \faulteq_{\fnot{P_f},\fnot{E}} t_2'$.
	
	Also, $m_1$ can not be an input to $d$. 
	Otherwise, we would know that $m_1 \in \omega_f'(t_1')$. 
	Since $\omega_f' \faulteq_{\fnot{P_f},\fnot{E}} \omega_c'$ there exists $m_2 \in \omega_c'(t_2')$ with $m_1 \faulteq_{\fnot{P_f},\fnot{E}} m_2$.
	By construction of $\faulteq_{\fnot{P_f},\fnot{E}}$ (Definition \ref{def:message:equiv}) $m_1$ and $m_2$ are communicated over the same port (Definition \ref{def:deterministic}), and since $d$ is input-indiscriminate $t_1' \concat m_1 \faulteq_{\fnot{P_f},\fnot{E}} t_2' \concat m_2$ and $\omega_c' \models d \trans{t_2' \concat m_2}$ (also Definition \ref{def:deterministic}).
	This means, again, $t_1'$ would not be the longest suffix with the property mentioned above.
	
	Since all CFTs of $d$ are correct w.r.t. $d$, also all clauses $P_i'$, defined by its CFTs, are correct w.r.t. $d$.
	Let $\fnot{P^i_j}$ and $\fnot{P^k_l}$ be arbitrary clauses for CFTs with output events $A^x_i$ and $A^y_l$ contained in the clause $P_f$.
	Due to correctness w.r.t. $d$, we know that there exist traces $t^a_2$ and $t^b_2$ such that $\omega_c' \models d \trans{t^a_2}$ and $\omega_c' \models d \trans{t^b_2}$ with $t_1 \faulteq_{\fnot{A^x_i}, \fnot{P^i_j}} t^a_2$ and $t_1 \faulteq_{\fnot{A^y_l}, \fnot{P^k_l}} t^b_2$.
	
	We inductively show that $t^a_2$ and $t^b_2$ are in fact the same traces, i.e. $t^a_2 = t^b_2$.
	
	Induction Hypothesis:
	For prefixes ${t^a}_2'$ of $t^a_2$ with length $n$, and all prefixes ${t^b}_2'$ of $t^b_2$ with length $n$ it holds that ${t^a}_2' = {t^b}_2'$.
	
	Induction Start:
	For length $n=0$ ${t^a}_2'$ and ${t^b}_2'$ are empty traces and the Hypothesis holds trivially.
	
	Induction Step:
	Let ${t^a}_2' = {t^a}_2'' \concat m^a_2$ and ${t^b}_2' = {t^b}_2'' \concat m^b_2$.
	
	Case 1:
	$m^a_2$ is input. Since $\omega_c'$ does not provide irrelevant inputs, $m^a_2$ has to be relevant w.r.t. $\faulteq_{\fnot{P_f},\fnot{E}}$, and analogous for $m^b_2$.
	Further, since $t^a_2 \faulteq_{\fnot{A^x_i}, \fnot{P^i_j}} t_1$ and $t^b_2 \faulteq_{\fnot{A^y_l}, \fnot{P^k_l}} t_1$, there exists an $m_1$ with $t_1' \concat m_1$ being a prefix of $t_1$ with ${t^a}_2' \concat m^a_2 \faulteq_{\fnot{A^x_i}, \fnot{P^i_j}} t_1' \concat m_1$ and ${t^b}_2' \concat m^b_2 \faulteq_{\fnot{A^y_l}, \fnot{P^k_l}} t_1' \concat m_1$.
	Therefore, $m_1, m^a_2, m^b_2$ are communicated on the same port.
	By construction, $\omega_c'$ only provides one relevant input per port, therefore $m^a_2 = m^b_2$, and thus with the induction hypothesis ${t^a}_2' \concat m^a_2 = {t^b}_2' \concat m^b_2$.
	
	Case 2: $m^a_2$ is output.
	Since $d$ is output deterministic and by induction hypothesis, $m^a_2 = m^b_2$, and again ${t^a}_2' \concat m^a_2 = {t^b}_2' \concat m^b_2$.
	
	Since $t^a_2 = t^b_2$ and by the structure of $P_f$ (Lemma \ref{lem:clause:complete}), by the definition of $\faulteq_{\fnot{A^x_i}, \fnot{P^i_j}}$, $\faulteq_{\fnot{A^y_l}, \fnot{P^k_l}}$, and $\faulteq_{\fnot{P_f},\fnot{E}}$ according to Definition \ref{def:message:equiv}, and since $t^1 \faulteq_{\fnot{A^x_i}, \fnot{P^i_j}} t^a_2$, and $t^1 \faulteq_{\fnot{A^y_l}, \fnot{P^k_l}} t^a_2$, it follows iteratively over all respective clauses in $P_f$ that $t^1 \faulteq_{\fnot{P_f}, \fnot{E}} t^a_2$.
	Therefore, there exists a $t_2$ which contradicts the assumption that $\omega_f', \omega_c', t_1$ is a counterexample for the safety of $d$, thus also $\omega_f, \omega_c, t_1$ is no counterexample.
	Finally, this shows that $(P_{c,d}',E)$ is correct w.r.t. $d$.
\end{proof}

Now, we have constructed a CFT $(P_{c,d}', E)$, which is correct w.r.t. $c$ and $d$.
We can show that $(P_{c,d}', E)$ is a also correct w.r.t. the composition o $c$ and $d$.

\begin{theorem}
	\label{thm:cft:composition:strict}
	Let $c$ and $d$ be components, $p_1 \ldots p_n$ ports, which are input ports of $c$ and outputs ports of $d$, $A^1_1 \ldots A^1_{m_1}, \ldots A^n_{1} \ldots A^n_{m_n}$ events with $A^i_j$ being an event on port $p_i$, i.e. an input event of $c$ and an output event of $d$.
	Let further be $(P_c,E)$ be a CFT of $c$ and $(P^i_j, A^i_j)$ the CFTs of $d$ for the output events $A^i_j$.
	
	Then the strict CFT composition $(P_{c,d}',E)$ (Definition \ref{def:cft:composition:strict}) is correct w.r.t. to the composition of $c$ and $d$. 
\end{theorem}
\begin{proof}
	From Lemmas \ref{lem:correct:c} and \ref{lem:correct:d} it follows that $(P_{c,d}',E)$ is correct w.r.t. $c$ and w.r.t. $d$.
	From Theorem 1 in \cite{greiner2016noninterference}, it then directly follows that $(P_{c,d}, E)$ is correct w.r.t. the composition of $c$ and $d$.
\end{proof}

It is left to show that correctness of the strictly composed CFT for the composition, implies that the composed CFT according to Definition \ref{def:cft:composition} is correct w.r.t. the composition of $c$ and $d$.
We prove Theorem \ref{thm:compositioncorrectness}.
\begin{proof}[Proof for Theorem \ref{thm:compositioncorrectness}]
When we compare the definition of strict CFT composition (Definition \ref{def:cft:composition:strict}) and CFT composition (Definition \ref{def:cft:composition}), we see that the strict composition keeps the events, on which the components are composed in the formula, while they disappear in the definition of non-strict composition.

Note that the normal composition therefore considers messages over the respective ports irrelevant.

We consider messages over these ports as outputs, therefore is the equivalence relation over messages implied by strict composition more strict that the equivalence relation over messages implied by composition.
Therefore every counterexample for $(P_{c,d},E)$ is also a counterexample for $(P_{c,d}',E)$

Therefore it directly follows that $(P_{c,d},E)$ is correct w.r.t. the composition of $c$ and $d$. 
\end{proof}

\renewcommand{\and}{and}

\end{document}